\documentclass[conference]{IEEEtran}
\usepackage{cite}
\usepackage{amsmath,amssymb,amsfonts,amsthm}
\usepackage{algorithmic}
\usepackage{graphicx}
\usepackage{textcomp}
\usepackage{xcolor}
\usepackage{comment}
\usepackage{float}

\usepackage[caption=false,font=footnotesize]{subfig}
\DeclareMathOperator*{\esssup}{esssup}
\newtheorem{theorem}{Theorem}
\newtheorem{lemma}[theorem]{Lemma}
\newtheorem{corollary}[theorem]{Corollary}
\newtheorem*{theorem*}{Theorem}
\newtheorem*{lemma*}{Lemma}
\newtheorem*{corollary*}{Corollary}
\newenvironment{sketch of proof}{\proof}{\endproof}
\DeclareMathOperator*{\argmin}{argmin}
\usepackage{physics}
\usepackage{MnSymbol}
\usepackage[bb=boondox]{mathalfa}
\DeclareMathOperator*{\argmax}{arg\,max}

\def\BibTeX{{\rm B\kern-.05em{\sc i\kern-.025em b}\kern-.08em
    T\kern-.1667em\lower.7ex\hbox{E}\kern-.125emX}}
\begin{document}

\title{Quickest Detection over Sensor Networks with Unknown Post-Change Distribution}

\author{\IEEEauthorblockN{Deniz Sargun, C. Emre Koksal}
\IEEEauthorblockA{\textit{Department of Electrical and Computer Engineering}\\
\textit{The Ohio State University}\\
Columbus, OH, USA\\
\{sargun.1, koksal.2\}@osu.edu}
}

\maketitle

\begin{abstract}
We propose a quickest change detection problem over sensor networks where both the subset of sensors undergoing a change and the local post-change distributions are unknown. Each sensor in the network observes a local discrete time random process over a finite alphabet. Initially, the observations are independent and identically distributed (i.i.d.) with known pre-change distributions independent from other sensors. At a fixed but unknown change point, a fixed but unknown subset of the sensors undergo a change and start observing samples from an unknown distribution. We assume the change can be quantified using concave (or convex) local statistics over the space of distributions. We propose an asymptotically optimal and computationally tractable stopping time for Lorden's criterion. Under this scenario, our proposed method uses a concave global cumulative sum (CUSUM) statistic at the fusion center and suppresses the most likely false alarms using information projection. Finally, we show some numerical results of the simulation of our algorithm for the problem described.
\end{abstract}


\section{Introduction} \label{section: introduction}
Detection of change in a random process has many applications. Most recently, quickest change detection (QCD), a specific formulation of the more general change detection problem, has been used in as diverse areas as behavioral economics \cite{inconsistent}, intruder detection \cite{intruder}, cognitive radios \cite{cognitive2}, covert communication \cite{covert}, energy harvesting sensor networks \cite{harvesting} and linear systems \cite{lowlinear}. The common objective in these works has been minimizing a well defined detection delay after a single unknown change point in the observation characteristics under the condition that false alarms do not occur frequently. Two important and well studied aspects of this problem that have helped extend the applicability of the theory of QCD over the seminal works of \cite{economic,continuous,procedures,optimal} have been the availability (or lack thereof) of pre- and post-change distributions to the detection algorithm and the single sensor vs sensor network formulations of the problem.

Change detection problems with unknown pre- or post change distribution have been investigated in \cite{bounds,minimax,data2,bin} after its first appearance in \cite{procedures} where the unknown post-change distribution belongs to a single parameter family. For problems with unknown parameters after the change point, \cite{bounds} has proposed the generalized likelihood ratio and mixture likelihood ratio stopping times that are asymptotically optimal as the average run length goes to infinity. Then, \cite{minimax} has (asymptotically) optimally solved Lorden's problem (a minimax QCD problem), Pollak's problem (a non-Bayesian conditional QCD problem) and the Bayesian formulation of QCD assuming the worst realization of pre- and post-change distributions under the assumption that least favorable distributions (LFDs) exist and are identifiable. In \cite{data2}, the authors have proposed a novel constraint on the frequency of observations before the change point while the post-change distribution is unknown up to a finite set of alternatives and have successively shown an asymptotically optimal algorithm satisfying this new condition. Finally, the authors of \cite{bin} use binning as another alternative method when pre-change distribution is known and is distinguishable from the unknown post-change distribution under some level of discretization.

Similarly, the (decentralized) sensor network problem has been proposed in \cite{decentralized} and has been studied in detail in \cite{decentralized2,information,efficient,censor}. In \cite{decentralized2}, a Bayesian formulation has been asymptotically solved where the sensors have restricted local memory but full feedback from the fusion center. The authors of \cite{information} have bound the asymptotic detection delay when the sensors have limited or full local memory while processing information to be sent to the fusion center using finite alphabets. The proposed algorithms have achieved asymptotic optimality by quantizing likelihood ratios locally to form messages for limited memory systems and sending local decisions of change for full local memory networks. Then, \cite{efficient} has proposed using the sum of local CUSUM statistics at the decision center when an unknown subset of sensors undergo change which has also been proven to be asymptotically optimal universally for any such unknown subset. Finally, the paper \cite{censor} has limited the frequency of using the channel from the sensors to the fusion center by sending local CUSUM statistics only if they are above a threshold and has proven asymptotic optimality even if sensors undergo change asynchronously.

We proposed a computationally efficient change detection algorithm that utilizes the knowledge of the pre-change distribution to suppress the most likely false alarms when the post-change distribution is unknown in \cite{separating-arxiv}. In the same paper, we have shown that our proposed method is asymptotically optimal up to a multiplicative constant for Lorden's problem. In this paper, we extend those results to sensor networks with an unknown subset of affected sensors and with unknown local distributions after the change point. To the best of our knowledge, this is the first approach to address a QCD problem over sensor networks with unknown local post-change distributions. We also prove asymptotic optimality under Lorden's criterion.

We are motivated from the fact that if the statistics that describe the change in the local observation process are concave (or convex) over the space of probability distributions, the fusion center can define a concave global CUSUM statistic that can identify the most likely false alarms to be suppressed via information projection. Further, this global statistic also preserves Lipschitz continuity. Satisfying these conditions, we can extend the information projection test for QCD over a single sensor \cite{separating-arxiv} to the QCD over sensor networks when local post-change distributions are unknown.

Our contributions can be summarized as follows.
\begin{enumerate}
    \item We introduce a novel framework to address the QCD problem over sensor networks with unknown post-change distributions.
    \item We do not assume a finite set of alternatives for the local post-change distributions nor the existence of LFDs. We allow an unknown subset of affected sensors and heterogeneity across the sensors, i.e. alphabets, pre-change and post-change distributions are not necessarily the same across sensors.
    \item We extend the idea to suppress most likely false alarms to the QCD problem over sensor networks utilizing the knowledge of the pre-change distribution and the rule for stopping.
    \item We show that our algorithm is computationally lighter than the corresponding generalized likelihood ratio test (GLRT) and prove that it is asymptotically optimal for Lorden's problem up to a multiplicative constant.
\end{enumerate}

\section{Model and Problem Statement} \label{section: model}
Assume there are \(J\) sensors monitoring physically distant locations or different variables of a single process. (ex. pressure, temperature and chemical concentration of a refinement process) For each sensor \(j=1,\dots,J\), let \(\mathcal{A}_j=\{a_{j,1},\dots,a_{j,m_j}\}\) denote the finite alphabet and \(\mathcal{P}_j\) denote the probability simplex of probability mass functions (p.m.f.s) \(f_j\) over \(\mathcal{A}_j\). The observation at sensor \(j\) at time \(k\) is denoted \(X_{j,k}\) where, before unknown change point \(t_1\), each random variable \(X_{j,k}\) is i.i.d. with known pre-change p.m.f \(f_{j,0}\). After \(t_1\), an unknown subset of sensors \(\mathcal{J}\subset\{1,\dots,J\}\) undergoes a change and for all \(j\in\mathcal{J}\) and \(k\geq t_1\) \(X_{j,k}\) are i.i.d. with the unknown post-change distribution\(f_{j,1}\), independent of previous observations \(X_{j,1},\dots,X_{j,t_1-1}\). Assume \(\mathcal{J}\neq\emptyset\) and that for each observation process there exists a statistic\footnote{Some examples for concave (or convex) and Lipschitz continuous \(q_j\) over \(\mathcal{P}_j=\{f_j\}\): (1) any expectation of the form \(\sum_{a_j}h(a_j)f_j(a_j)\) like the \(r\)th moment or cross entropy \(H(f_j,g_j)\), (2) variance, (3) entropy, (4) operations with the previous examples that preserve concavity (or convexity) like \(I(f_j\|g_j)\)} \(q_j:\left(\mathcal{P}_j,l_1\right)\to\mathbb{R}\) concave and Lipschitz continuous with Lipschitz constant \(L_j\) that satisfies \(q_j(f_{j,0})=q_{j,0}=0<\underline{q}_j\leq q_{j,1}=q_j(f_{j,1})\). Thus, each \(f_{j,1}\) is sampled from an \textit{unknown subset} \(\mathcal{F}_{j,1}\) of \(q_j^{-1}([\underline{q}_j,\infty))\). We use \([J]\) to denote \(\{1,\dots,J\}\), \(X_k\) to denote \((X_{1,k},\dots,X_{J,k})\), \(X_{j,k}^l\) to denote \((X_{j,k},X_{j,k+1},\dots,X_{j,l})\), \(X_k^l\) to denote \((X_{1,k}^l,\dots,X_{J,k}^l)\) and \(\hat{f}_{X_{j,k}^l}\) and \(\hat{f}_{X_k^l} \) for the empirical p.m.f.s of \(X_{j,k}^l\) and \(X_k^l\), i.e. \(\hat{f}_{X_{j,k}^l}(x)=\frac{1}{l-k+1}\sum_{k'=k}^l\mathbb{1}_{X_{j,k'}}(x)\) and \(\hat{f}_{X_k^l}(x)=\frac{1}{l-k+1}\sum_{k'=k}^l\mathbb{1}_{X_{k'}}(x)\). Finally, \(I\left(f\middle\|f'\right)=\sum_a f(a)\log\frac{f(a)}{f'(a)}\) denotes the Kullback-Leibler (KL) divergence between distributions \(f\) and \(f'\) in nats and \(P_{f_1,\mathcal{J},t_1}\) denotes the probability law when post-change distributions are \(f_{1,1},\dots,f_{J,1}\), subset of affected sensors is \(\mathcal{J}\) and change point is \(t_1\). When change does not occur, we denote \(P_{f_1,\mathcal{J},\infty}\) briefly as \(P_\infty\).

We want to solve Lorden's problem of minimizing the worst average detection delay (WADD) subject to a minimum average run length (ARL) in a sensor network. We define the WADD and the ARL of a stopping time \(t_a\) as follows.
\begin{align}
    \inf_{t_a}\ &\sup_{f_1,\mathcal{J},t_1}\esssup_{X_1^{t_1-1}} E_{f_1,\mathcal{J},t_1}\left(\left(t_a-t_1+1\right)^+\middle|X_1^{t_1-1}\right) \label{equation: wadd}\\
    \text{s.t.}\ &E_\infty t_a\geq\gamma. \label{equation: arl}
\end{align}

\section{Network Information Projection Test} \label{section: algorithm}

\subsection{Global statistic}
Let \(\mathcal{A}=\times_j\mathcal{A}_j\) be the Cartesian product alphabet with \(m=\prod_jm_j\) letters and \(\mathcal{P}\) denote the probability simplex over \(\mathcal{A}\). Then, given the change point \(t_1\), for \(k<t_1\), \(X_k\) is i.i.d. with \(f_0=\otimes_j f_{j,0}\in\mathcal{P}\) and for \(k\geq t_1\), \(X_k\) is i.i.d. with \(f_1=\otimes_j\tilde{f}_j\in\mathcal{P}\) where \(\tilde{f}_j=f_{j,1}\in\mathcal{F}_{j,1}\) if \(j\in\mathcal{J}\) and \(f_{j,0}\) otherwise. Let us define a global statistic \(q:\mathcal{P}\to\mathbb{R}\) to study the effect of the change on the network.
\begin{align*}
    q(f)&= \sum_jq_j\left(f_j\right)\\
    f_j(a_j)&= \sum_{a_1}\dots\sum_{a_{j-1}}\sum_{a_{j+1}}\dots\sum_{a_J}f(a_1,\dots,a_j,\dots,a_J)\\
    &= \sum_{-a_j}f(a).
\end{align*}
Note that \(q\) is concave and Lipschitz continuous with constant \(L=\sum_j L_j\) over \(\mathcal{P}\). (see Appendix \ref{subsection: concavity} and \ref{subsection: lipschitzness}) Thus, for any \(\eta\), \(\left\{f\in\mathcal{P}\middle|q(f)=\sum_jq_j(f_j)\geq\eta\right\}\) is closed, convex and bounded.

\subsection{Algorithm}
Using \(q\), we define a CUSUM statistic to be utilized in the first stage of detection. Our CUSUM statistic is similar to the SUM scheme in \cite{efficient} where this problem is solved by computing a local CUSUM statistic and then summing them in the fusion center, i.e. \(\sum_j\max_{l_j\leq k+1}(k-l_j+1)q_j\left(\hat{f}_{j,l_j}^k\right)\). Although the local maximization over the possible change points reduces the noise in the fusion center, the algorithm proposed in \cite{efficient} may be suboptimal if multiple sensors are affected simultaneously since the estimated change points are computed locally. Therefore we employ a single maximization at the fusion center. Using a sum of sign changing drift terms for each of the local sensors (ex. \(\sum_j kq_j\left(\hat{f}_{j,1}^k\right)\)), we want the overall drift to become positive even if only a small subset \(\mathcal{J}\subset[J]\) of sensors undergo a change, ex. if \(\mathcal{J}\) is a singleton. Therefore we use local statistics that have no drift before the change point and add a negative drift of \(-\kappa\) at the fusion center. Let \(S_k\) denote the global CUSUM statistic, then
\begin{align*}
    S_k&= \max_{\tau_k\leq l\leq k+1}(k-l+1)\left(\sum_jq_j\left(\hat{f}_{X_{j,l}^k}\right)-\kappa\right)\\
    &= \max_{\tau_k\leq l\leq k+1}(k-l+1)\left(q\left(\hat{f}_{X_l^k}\right)-\kappa\right)
\end{align*}
where \(0<\kappa<\min_j\underline{q}_j\) and we denote \(\min_j\underline{q}_j-\kappa=\underline{q}\). Given a window size \(n\), we define the most likely false alarm distribution \(f_n^*\), i.e. the empirical distribution that is most likely under \(f_0\) among those for which \(S_k\geq c^S\), and KL divergence \(D_k\) from it as
\begin{align*}
    f_n^*&= \argmin_{q(f)\geq\frac{c^S}{n}+\kappa}I\left(f\middle\|f_0\right)\\
    i_k&= \argmax_{\tau_k\leq l\leq k+1}(k-l+1)\left(q\left(\hat{f}_{X_l^k}\right)-\kappa\right)\\
    n_k&= k-i_k+1\\
    D_k&= I\left(\hat{f}_{X_{i_k}^k}\middle\|f_{n_k}^*\right)
\end{align*}
Note that if \(\{q(f)\geq\frac{c^S}{n}+\kappa\}\neq\emptyset\), then \(f_n^*\) is unique since since \(I(f\|f_0)\) is strictly convex and continuous in \(f\). Finally, define the reset times \(\tau_k\) and the stopping times \(t_S\) and \(t_I\) as
\begin{align*}
    \tau_{k+1}&= \begin{cases}
    k+1, &S_k\geq c^S, D_k<c_{n_k}^D\\
    \tau_k, &\text{otherwise} 
    \end{cases}\\
    t_S&= \inf\left\{k\middle|S_k\geq c^S\right\}\\
    t_I&= \inf\left\{k\middle|\left(S_k,D_k\right)\geq\left(c^S,c_{n_k}^D\right)\right\}.
\end{align*}
We call \(t_I\) as the network information projection test (NIPT) and prove that it is asymptotically optimal as \(\gamma\to\infty\) up to a constant factor.

\subsection{Complexity} \label{subsection: complexity}
The algorithm can be initialized by computing the most likely false alarm distributions \(f_n^*\) for \(n\in\mathcal{O}\left(c^S\right)\). For each \(n\), \(f_n^*\) is the solution to the convex optimization
\begin{align*}
    \min_f\ &I\left(f\middle\|f_0\right)\\
    \text{s.t.}\ &{q(f)\geq\frac{c^S}{n}+\kappa}\\
    &f\in\mathcal{P}.
\end{align*}
Without loss of generality, assume the set of feasible solutions is nonempty, \(f_0\) is not feasible and \(f_0(a)>0\) for all \(a\). For differentiable \(q\), using the KKT conditions, we can reformulate this as a search over Lagrangian parameters \(\lambda\geq 0\) and \(\nu\) that satisfy
\begin{align*}
    \sum_af_0(a)\exp(\lambda\tilde{q}(f))\exp(-\nu)&= 1\\
    q(f)&= \frac{c^S}{n}+\kappa
\end{align*}
where \(\tilde{q}\) is the gradient of \(q\) and \(f(a)=f_0(a)\exp(\lambda\tilde{q}(f)-\nu)\). Then, if, at each iteration, \(q\) and \(\tilde{q}\) can be computed in \(\mathcal{O}(m)\) time from \(f\), solving \(f_n^*\) has complexity \(\mathcal{O}\left(m\log\frac{1}{\epsilon}\right)\) where \(\epsilon\) is the maximum allowable \(l_1\) error for the solution. Then, the initialization process for NIPT has complexity \(\mathcal{O}\left(c^Sm\log\frac{1}{\epsilon}\right)\). (or, as will be shown, \(\mathcal{O}\left(\log\gamma\ m\log\frac{1}{\epsilon}\right)\)) Finally, the run time complexity to compute \(S_k\) and \(D_k\) is \(\mathcal{O}(\gamma m)\).

On the other hand, GLRT has no initialization but solves the noniterative optimization
\begin{align*}
    \max_{l\leq k+1,\mathcal{J}}\sup_{f_1}\sum_{k'=l}^k \sum_{j\in\mathcal{J}}\log\frac{f_1(X_k)}{f_0(X_k)}
\end{align*}
during run time. Note that, since the set of local post-change distributions are unknown, GLRT can only utilize the knowledge of the super sets \(q_j^{-1}([\underline{q}_j,\infty))\) instead of \(\mathcal{F}_{j,1}\). Then, at each iteration \(k\), GLRT solves the convex problem
\begin{align*}
    \max_{l\leq k+1,\mathcal{J}}\min_{f_{j,1}|j\in\mathcal{J}}\sum_{j\in\mathcal{J}}I\left(\hat{f}_{X_{j,l}^k}\middle\|f_{j,1}\right)
\end{align*}
in \(\Omega\left(k2^Jm\log\frac{1}{\epsilon}\right)\) time, using, for example, the center of gravity method \cite{optimization}. Then, the run time complexity of GLRT is larger than that of the NIPT \(\mathcal{O}(\gamma m)\).

\section{Quickest Change Detection over Sensor Networks} \label{section: bounds}
We prove the asymptotic optimality by lower bounding ARL and upper bounding WADD as a function of the thresholds \(c^S\) and \(c_n^D\). For the full proofs of Theorem \ref{theorem: arl} and Lemma \ref{lemma: add}, see Appendix B of \cite{quickest-arxiv}.

First, we prove a lemma that characterizes the stopping time \(t_I\) in terms of the stopping time \(t_S\) and the conditional expected probability that the second stage of the algorithm classifies the empirical distribution as change, i.e. the probability that \(D_k\geq c_{n_k}^D\) at \(k=t_S\).
\begin{lemma} \label{lemma: conditional}
For any \(f_1,\mathcal{J}\) and \(t_1\),
\begin{align*}
    E_{f_1,\mathcal{J},t_1} t_I&= \frac{E_{f_1,\mathcal{J},t_1} t_S}{E_{f_1,\mathcal{J},t_1}\left(P_{f_1,\mathcal{J},t_1}\left(D_{t_S}\geq c_{n_{t_S}}^D\middle|t_S\right)\right)}.
\end{align*}
\end{lemma}
\begin{proof}
Let us express \(E_{f_1,\mathcal{J},t_1} t_I\) conditional on \(t_S\). Whenever \(S_k\) crosses the threshold \(c^S\), the algorithm stops and declares change at \(t_I=t_S=k\) if and only if \(D_k\geq c_{n_k}^D\). Otherwise, \(D_k<c_{n_k}^D\) and \(\tau_{k+1}=k+1\) resetting the algorithm. Thus,
\begin{align*}
    E t_I&= E\left(E\left(t_I\middle|t_S\right)\right)\\
    &= E t_S+E\left(P\left(t_I\neq t_S\middle|t_S\right)\right)E t_I\\
    &= \frac{E t_S}{E\left(P\left(D_{t_S}\geq c_{n_{t_S}}^D\middle|t_S\right)\right)}
\end{align*}
where we have briefly used \(E\) and \(P\) to denote \(E_{f_1,\mathcal{J},t_1}\) and \(P_{f_1,\mathcal{J},t_1}\) respectively.
\end{proof}

\subsection{Average run length}
Using Lemma \ref{lemma: conditional}, we show that the ARL increases asymptotically exponentially with the first threshold \(c^S\) at a rate that approaches \(v^*+\frac{2\kappa}{L^2}\) for a suitable choice of \(c_n^D\).
\begin{theorem} \label{theorem: arl}
For any \(\rho\in(0,1)\), if \(c_n^D=c^D\) for \(n>(1-\rho)\frac{c^S}{\kappa}\) and \(c^D\geq\frac{(2-\rho)^2\kappa^2}{2(1-\rho)^2L^2}\), then
\begin{align*}
    ARL(t_I)&\geq \exp\left(\left(v^*+\frac{2\kappa}{L^2}+o(1)\right)c^S\right)
\end{align*}
as \(c^S\to\infty\) and \(\rho\to 0\) where \(v^*>0\) satisfies
\begin{align}
    \psi(v^*)&= \log E_\infty\exp(v^*\left(\left(\sum_jX_{j,k}\right)-\kappa\right))=0. \label{equation: log moment generating function}
\end{align}
\end{theorem}
\begin{sketch of proof}
We use Lemma \ref{lemma: conditional} to bound the ARL of \(t_I\) in terms of ARL of \(t_S\) and the probability that the alarm is suppressed and the algorithm is reset.

We first express the probability of reset conditional on the window size \(n\) at the first stopping time, \(t_S\), being less than a fraction of its most likely outcome. Using the \(l_1\) bound in \cite{inequalities}, we show that it is highly likely that \(n>(1-\rho)\frac{c^S}{\kappa}\). In that case, given \(c^D\), we bound the probability of reset using the Pythagorean theorem for relative entropy \cite{elements}. Finally, we lower bound the ARL of \(t_S\) with a suitable version of Wald's identity \cite{discrete}. Then, for large enough \(c^D\) the result follows.
\end{sketch of proof}

\subsection{Worst average detection delay}
In this subsection, we prove an asymptotic upper bound for the WADD in \eqref{equation: wadd}. First, similar to \cite{procedures}, we argue that the worst change point is at \(t_1=1\) where, in our case, we also have to account for the separate stages.
\begin{lemma} \label{lemma: change point}
For any \(f_1,\mathcal{J},t_1\) and \(X_1^{t_1-1}\), the stopping time \(t_I\) satisfies
\begin{align*}
    E_{f_1,\mathcal{J},t_1}\left(\left(t_I-t_1+1\right)^+\middle|X_1^{t_1-1}\right)&\leq E_{f_1,\mathcal{J},1}\left(t_S+t_I\right).
\end{align*}
\end{lemma}
\begin{sketch of proof}
This proof parallels to the proof of Lemma 6 in \cite{separating-arxiv}.

We first prove that \(t_I\) is upper bounded by the sum of the first stopping time \(t_S\) after the change the change point \(t_1\) and the restarted stopping time \(t_I\) that starts sampling after \(t_S\). Then, we use this bound in the expression of WADD \eqref{equation: wadd}. Since restarted stopping times are independent of past samples, we can set \(t_1=1\) and disregard the worst realization over samples before the change point.
\end{sketch of proof}
We prove another lemma to show that \(t_I\) is asymptotically linear with the threshold \(c^S\) if \(t_1=1\) and under certain conditions for \(c_n^D\) which allow empirical distributions sampled from a \(f_1\) to be classified as change with high probability.
\begin{lemma} \label{lemma: add}
For any \(f_1,\mathcal{J}\) and \(\rho>0\), if \(t_1=1\) and \(c_n^D=0\) for \(n\leq(1+\rho)\frac{c^S}{\underline{q}}\), then
\begin{align*}
    E_{f_1,\mathcal{J},1}t_I&\leq \frac{c^S}{\underline{q}}
\end{align*}
as \(c^S\to\infty\) and \(\rho\to 0\).
\end{lemma}
\begin{sketch of proof}
We use Lemma \ref{lemma: conditional} to express the detection delay of \(t_I\) when \(t_1=1\) in terms of the detection delay of \(t_S\) and the probability that the alarm is suppressed and the algorithm is reset.

Given \(f_1,\mathcal{J}\) and \(t_1=1\), we upper bound the delay of \(t_S\) bounding the probability that \(S_k\) has not crossed the threshold \(c^S\) even though it has a positive drift of \(\sum_{j\in\mathcal{J}}q_{j,1}-\kappa\geq\min_j\underline{q}_j-\kappa=\underline{q}\). Then, we lower bound the probability of reset using the fact that if with high probability \(t_S\leq(1+\rho)\frac{c^S}{\underline{q}}\), then \(c_n^D=0\) in which case the second stage of the algorithm, comparing \(D_k\) with \(c_n^D\), does not suppress any alarm.
\end{sketch of proof}
The next theorem established our upper bound for WADD using Lemma \ref{lemma: change point}, \ref{lemma: add} and the bound \(t_S\leq t_I\).
\begin{theorem} \label{theorem: wadd}
For any \(f_1,\mathcal{J},t_1\) and \(\rho>0\), if \(c_n^D=0\) for \(n\leq(1+\rho)\frac{c^S}{\underline{q}}\), then
\begin{align*}
    WADD(t_I)&\leq \frac{2c^S}{\underline{q}}
\end{align*}
as \(c^S\to\infty\) and \(\rho\to 0\).
\end{theorem}
\begin{proof}
For any \(f_1,\mathcal{J}\) and \(t_1\),
\begin{align*}
    &WADD(t_I)\\
    &= \sup_{f_1,\mathcal{J},t_1}\esssup_{X_1^{t_1-1}} E_{f_1,\mathcal{J},t_1}\left(\left(t_I-t_1+1\right)^+\middle|X_1^{t_1-1}\right)\\
    &\leq \sup_{f_1,\mathcal{J}}E_{f_1,\mathcal{J},1}\left(t_S+t_I\right)\\
    &\leq \frac{2c^S}{\underline{q}}
\end{align*}
where we have used Lemma \ref{lemma: change point} and the fact that \(t_S\leq t_I\).
\end{proof}
\begin{corollary}
For any \(f_1,\mathcal{J},t_1\) and \(0<\kappa<\frac{1}{2}\min_j\underline{q}_j\), there exists \(\rho\in(0,1)\) such that \((1+\rho)\frac{c^S}{\underline{q}}\leq(1-\rho)\frac{c^S}{\kappa}\) and if \(c_n^D=0\) for \(n\leq(1+\rho)\frac{c^S}{\underline{q}}\), \(c_n^D=c^D\) for \(n>(1-\rho)\frac{c^S}{\kappa}\) and \(c^D\geq\frac{(2-\rho)^2\kappa^2}{2(1-\rho)^2L^2}\), then
\begin{align}
    WADD(t_I)&\leq \frac{\log\gamma}{\underline{q}\left(\frac{v^*}{2}+\frac{\kappa}{L^2}\right)} \label{equation: bound}
\end{align}
as \(\gamma\to\infty\) and \(\rho\to 0\).
\end{corollary}
\begin{proof}
For \(0<\kappa<\frac{1}{2}\min_j\underline{q}_j\), \(\underline{q}=\min_j\underline{q}_j-\kappa>\kappa\) and we can choose \(\rho\in(0,1)\) such that \(\frac{1+\rho}{1-\rho}\leq\frac{\underline{q}}{\kappa}\).

Let \(c_n^D=\frac{(2-\rho)^2\kappa^2}{2(1-\rho)^2L^2}\) for \(n>(1-\rho)\frac{c^S}{\kappa}\) and \(c_n^D=0\) for \(n\leq(1+\rho)\frac{c^S}{\underline{q}}\). Then, by Theorem \ref{theorem: arl} and \ref{theorem: wadd}, there exist \(c^S\) and \(\rho\) such that
\begin{align*}
    ARL(t_I)&\sim \exp\left(\left(v^*+\frac{2\kappa}{L^2}\right)c^S\right)\geq\gamma\\
    WADD(t_I)&\sim \frac{2c^S}{\underline{q}}.
\end{align*}
Thus, for \(c^S\geq\frac{\log\gamma}{\left(v^*+\frac{2\kappa}{L^2}\right)}\) we have \(WADD(t_I)\sim\frac{2\log\gamma}{\underline{q}\left(v^*+\frac{2\kappa}{L^2}\right)}\).
\end{proof}

\section{Numerical Results} \label{section: comparison}
In this section we give a numerical example by describing the problem and comparing the empirical detection performance with our bound in \eqref{equation: bound}.

Consider the finite alphabet \(\left\{-4,\dots,4\right\}\). Over a sensor network with 3 nodes, we let the local pre-change distributions \(f_{j,0}\) to be zero mean discrete Gaussian with \(\sigma^2=1\), i.e. \(f_{j,0}(a)=C\exp\left(-\frac{a^2}{2d^2}\right)\) where \(d\sim 1\). The unknown post-change distributions \(f_{j,1}\) satisfy \(\sigma^2\geq\underline{\sigma}^2=2\). For \(q_j(f_j)=\sigma_{f_j}^2-1\), \(q_j\)s are concave over \(\mathcal{P}_j\), thus, we can utilize the NIPT. We randomly specify \(1000\) different post-change distributions from \(\mathcal{F}_{j,1}=\{\sigma^2\geq 2\}\), let \(t_1=1,3,10,30,100,300\) or \(\infty\) and average over \(1000\) realizations for ARL and \(180k\) for WADD.
\begin{figure}[H]
    \centering
    \includegraphics[scale=0.205]{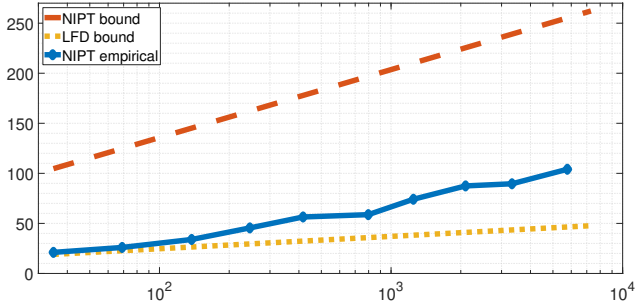}
    \caption{The WADD vs ARL for \(|J|=3,m=729\) and a change in the local variance}
    \label{figure: empirical}
\end{figure}
In Fig. \ref{figure: empirical} we have shown the empirical worst delay for NIPT over a range of ARLs and the asymptotic bound we have described in \eqref{equation: bound}. Since, for this example, we can identify local LFDs \cite{minimax}, we have also given an asymptotic lower bound for the delay when \(\mathcal{J}\) is known. The NIPT bound plotted may not be a strict bound because in general (1) Theorem \ref{theorem: wadd} loosely bounds \(t_S\) with \(t_I\) and for this example (2) the random walk \(S_k\) is bounded with \(\sum_jX_j^2-(\kappa-J)\) to find a suitable \(v^*\) and (3) we have restricted the change point \(t_1\) for testing. Comparing the two bounds, we can conclude that, for this problem, NIPT is asymptotically optimal up to a multiplicative constant less than \(6\).

\section{Conclusion} \label{section: conclusion}
We have proposed a novel QCD problem and a computationally feasible solution over sensor networks where sensor level distributions are unknown after the change point. The novelty in this paper is that we extend the definition of unknown post-change distributions as referred in QCD problems over sensor networks \textit{from} just assuming an unknown subset of affected sensors \textit{to} assuming unknown subset of affected sensors \textit{and} unknown local distributions after change.

We have proven the asymptotic optimality of network information projection testing for Lorden's criteria up to a multiplicative constant. Under some constraints on \(c_n^D\) and \(\kappa\), our two stage algorithm first performs similar to a CUSUM method with its drift changing property and then compares the empirical distribution to a likely false alarm given the threshold at the first stage and the number of samples that lead to an alarm. We have also given a numerical example of the performance of the algorithm for a special case of the problem.

As future work, we want to find tighter bounds, especially in Theorem \ref{theorem: wadd}, to prove asymptotic optimality with the optimum constant. Our algorithm is not fully decentralized since it requires \(\hat{f}\) at the fusion center whenever the first threshold is crossed, \(S_k\geq c^S\), to compute \(D_k\). Thus, we also want to construct a fully decentralized asymptotically optimal algorithm for the QCD problem over networks with unknown local post-change distributions. One method we are considering for decentralization is by modifying the second stage of our algorithm with \(c^D=\infty\). Finally, we also want to investigate the effect of having different change points across sensors that get affected and that of quantization for distribution over continuous alphabets.

\appendices
\section{Properties of \(\sum_jq_j\)}
\subsection{Concavity} \label{subsection: concavity}
For any \(f,g\in\mathcal{P}\) and \(\alpha\in[0,1]\),
\begin{align*}
    &q\left(\alpha f+(1-\alpha)g\right)\\
    &= \sum_jq_j\left(\sum_{-j}\alpha f(a)+(1-\alpha)g(a)\right)\\
    &\geq \sum_j\alpha q_j\left(\sum_{-j}f(a)\right)+(1-\alpha)q_j\left(\sum_{-j}g(a)\right)\\
    &= \alpha\sum_jq_j\left(\sum_{-j}f(a)\right)+(1-\alpha)\sum_jq_j\left(\sum_{-j}g(a)\right)\\
    &= \alpha q(f)+(1-\alpha)q(g)
\end{align*}
and \(q\) is concave over \(\mathcal{P}\).

\subsection{Lipschitzness} \label{subsection: lipschitzness}
For any \(f,g\in\mathcal{P}\), \(q\) is Lipschitz continuous over \(\mathcal{P}\) with Lipschitz constant \(\sum_jL_j\).
\begin{align*}
    \left|q(f)-q(g)\right|&= \left|\sum_jq_j\left(\sum_{-j}f(a)\right)-\sum_jq_j\left(\sum_{-j}g(a)\right)\right|\\
    &\leq \sum_j\left|q_j\left(\sum_{-j}f(a)\right)-q_j\left(\sum_{-j}g(a)\right)\right|\\
    &\leq \sum_jL_j\norm{\sum_{-j}f(a)-\sum_{-j}g(a)}_1\\
    &= \sum_jL_j\sum_{a_j}\left|\sum_{-a_j}f(a)-g(a)\right|\\
    &\leq \sum_jL_j\sum_a\left|f(a)-g(a)\right|\\
    &= \sum_jL_j\norm{f-g}_1\\
    &= \left(\sum_jL_j\right)\norm{f-g}_1
\end{align*}

\section{Proofs of Theorems}
\subsection{Proof of Theorem \ref{theorem: arl}}
\begin{theorem*}
For any \(\rho\in(0,1)\), if \(c_n^D=c^D\) for \(n>(1-\rho)\frac{c^S}{\kappa}\) and \(c^D\geq\frac{(2-\rho)^2\kappa^2}{2(1-\rho)^2L^2}\), then
\begin{align*}
    ARL(t_I)&\geq \exp\left(\left(v^*+\frac{2\kappa}{L^2}\right)c^S\right)
\end{align*}
as \(c^S\to\infty\) and \(\rho\to 0\) where \(v^*>0\) satisfies
\begin{align*}
    \psi(v^*)&= \log E_\infty\exp(v^*\left(\left(\sum_jX_{j,k}\right)-\kappa\right))=0.
\end{align*}
\end{theorem*}
\begin{proof}
From Lemma \ref{lemma: conditional},
\begin{align}
    E_\infty t_I&= \frac{E_\infty t_S}{E_\infty\left(P_\infty\left(D_{t_S}\geq c_{n_{t_S}}^D\middle|t_S\right)\right)}. \label{equation: arl conditional}
\end{align}
Given \(\rho\in(0,1)\), let \(N\) denote \((1-\rho)\frac{c^S}{\kappa}\). Then, for the denominator in \eqref{equation: arl conditional},
\begin{align}
    P_\infty\left(D_{t_S}\geq c_{n_{t_S}}^D\middle|t_S\right)&\leq P_\infty\left(n_{t_S}\leq N\middle|t_S\right) \nonumber\\
    &+ P_\infty\left(n_{t_S}>N,D_{t_S}\geq c_{n_{t_S}}^D\middle|t_S\right) \label{equation: arl denominator}
\end{align}
For the first term in \eqref{equation: arl denominator},
\begin{align*}
    P_\infty\left(n_{t_S}\leq N\middle|t_S\right)&\leq N\max_{n\leq N}P_\infty\left(n_{t_S}=n\middle|t_S\right)\\
    &\leq N\max_{n\leq N}P_\infty\left(q\left(\hat{f}_{X_{t_S-n+1}^{t_S}}\right)\geq\frac{c^S}{n}+\kappa\middle|t_S\right)\\
    &\leq N\max_{n\leq N}P_\infty\left(\norm{\hat{f}_n-f_0}_1\geq\frac{\frac{c^S}{n}+\kappa}{L}\middle|t_S\right)
\end{align*}
where we have used the Lipschitz continuity of \(q\) at \(f_0\). Using the \(l_1\) inequality in \cite{inequalities},
\begin{align*}
    P_\infty\left(n_{t_S}\leq N\middle|t_S\right)&\leq N\max_{n\leq N}2^m\exp\left(-\frac{n}{2}\left(\frac{\frac{c^S}{n}+\kappa}{L}\right)^2\right)\\
    &\leq 2^mN\exp\left(-\frac{(2-\rho)^2\kappa}{2(1-\rho)L^2}c^S\right)
\end{align*}
since \(n\leq N<\frac{c^S}{\kappa}\). For the second term in \eqref{equation: arl denominator}, if \(c_n^D=c^D\),
\begin{align*}
    &P_\infty\left(n_{t_S}>N,D_{t_S}\geq c_{n_{t_S}}^D\middle|t_S\right)\\
    &= \sum_{n>N}P_\infty\left(n_{t_S}=n\middle|t_S\right)P_\infty\left(D_{t_S}\geq c^D\middle|n_{t_S}=n,t_S\right)\\
    &\leq \sum_{n>N}P_\infty\left(I\left(\hat{f}_{X_{t_S-n+1}^{t_S}}\middle\|f_n^*\right)\geq c^D\middle|n_{t_S}=n,t_S\right)\\
    &\leq \sum_{n>N}P_\infty\left(I\left(\hat{f}_n\middle\|f_0\right)\geq c^D\middle|n_{t_S}=n,t_S\right)\\
    &\leq \sum_{n>N}(n+1)^m\exp\left(-nc^D\right)
\end{align*}
where we have used the Pythagorean theorem for relative entropy and an upper bound on the cardinality of the set of empirical distributions \cite{elements}. Then,
\begin{align*}
    P_\infty\Big(n_{t_S}>N,D_{t_S}\geq c_{n_{t_S}}^D\Big|t_S&\Big)\leq \sum_{n>N}\left(N+1\right)^{\frac{n}{N}m}\exp\left(-nc^D\right)\\
    &= \frac{(N+1)^m\exp(-Nc^D)}{1-\exp\left(-c^D+\frac{m}{N}\log(N+1)\right)}
\end{align*}
since \(n+1\leq\left(N+1\right)^{\frac{n}{N}}\) for \(n>N\) where we also assumed \(c^D>\frac{m}{N}\log(N+1)=\frac{m\log((1-\rho)\frac{c^S}{\kappa}+1)}{(1-\rho)\frac{c^S}{\kappa}}\).

Next, we bound lower bound the nominator in \eqref{equation: arl conditional}, \(E_\infty t_S\). For simplicity, assume \(q_j\)s are the mean operators, then \(S_k\) is the random walk with i.i.d. steps \(\sum_j X_{j,k}-\kappa\) that is bounded by zero below and above by \(c^S\). For any other function \(h\) over \(\mathcal{A}_j\) that does not yield the mean operator, we can modify the tilt parameter used in the proof of Wald's identity to get \(q_j(f_j)=E_{f_j}h(X_j)\) as in the proof in \cite{discrete}. Denote the number of zero crossings before \(t_S\) as \(N_0\). Then,
\begin{align*}
    E_\infty t_S&= E_\infty E_\infty\left(t_S\middle|N_0\right)\\
    &\geq E_\infty N_0.
\end{align*}
Since the steps are i.i.d. and bounded, moments of all order exist and we can apply Wald's identity to \(S_k\). Let \(k^*\) denote the latest upper or lower threshold crossing time after a reset. Then,
\begin{align*}
    P_\infty\left(S_{k^*}\geq c^S\right)&\leq \exp\left(-v^*c^S\right)
\end{align*}
and \(E_\infty t_S\geq\exp(v^*c^S)\). Finally we obtain,
\begin{align*}
    E_\infty t_I&\geq \frac{\exp(v^*c^S)}{2^mN\exp\left(-\frac{(2-\rho)^2\kappa}{2(1-\rho)L^2}c^S\right)+\frac{(N+1)^m\exp(-Nc^D)}{1-\exp\left(-c^D+\frac{m}{N}\log(N+1)\right)}}.
\end{align*}
If \(c^D\geq\frac{(2-\rho)^2\kappa^2}{2(1-\rho)^2L^2}\), the first term in the denominator dominates and \(E_\infty t_I\geq \exp\left(\left(v^*+\frac{2\kappa}{L^2}+o(1)\right)c^S\right)\) as \(c^S\to\infty,\rho\to 0\).
\end{proof}

\subsection{Proof of Lemma \ref{lemma: change point}}
\begin{lemma*}
For any \(f_1,\mathcal{J},t_1\) and \(X_1^{t_1-1}\), the stopping time \(t_I\) satisfies
\begin{align*}
    E_{f_1,\mathcal{J},t_1}\left(\left(t_I-t_1+1\right)^+\middle|X_1^{t_1-1}\right)&\leq E_{f_1,\mathcal{J},1}\left(t_S+t_I\right).
\end{align*}
\end{lemma*}
\begin{proof}
Let us define \(t_S^{(u)}\) as the \(u\)th time \(S_k\) crosses \(c^S\) and \(t_S^v\) and \(t_I^v\) as the the stopping times applied on the samples starting at \(X_v\). Then, \(t_S^{(0)}=0\) and for all \(u\geq 1\), \(t_S^{(u)}=t_S^{(u-1)}+t_S^{t_S^{(u-1)}+1}\).

We first show that for \(u^*=\inf\left\{u\middle|t_S^{(u)}\geq t_1\right\}\),
\begin{align}
    t_I\leq t_1-1+t_S^{t_1}+t_I^{t_S^{(u^*)}+1}. \label{equation: stopping time bound}
\end{align}

Let \(u\) be an arbitrary nonnegative integer. Without loss of generality, assume \(t_I\geq t_S^{(u)}\). Then, either \(t_I=t_S^{(u)}\) or the algorithm is restarted and \(t_I=t_S^{(u)}+t_I^{t_S^{(u)}+1}\). Therefore, we have for all \(u\), \(t_I\leq t_S^{(u)}+t_I^{t_S^{(u)}+1}\).

Similarly, let \(u\) and \(t_S^{(u)}<k\leq t_S^{(u+1)}\) be arbitrary. If \(\tau_{k+t_S^k-1}>k\), then there exists \(k'\) such that \(k<k'\leq k+t_S^k-1\) and \(\tau_{k+t_S^k-1}=\tau_{k'}=k'\). Thus, \(t_S^{(u+1)}\leq k'\leq k+t_S^k-1\). Else, \(\tau_{k+t_S^k-1}\leq k\)
\begin{align*}
    S_{k+t_S^k-1}&= \max_{\tau_{k+t_S^k-1}\leq i\leq k+t_S^k}Q(i,k)\\
    &\leq \max_{k\leq i\leq k+t_S^k}Q(i,k)\\
    &= S_{t_S^k}^k\geq c^S.
\end{align*}
Thus, \(t_S^{(u+1)}\leq k+t_S^k-1\). Setting \(u+1=u^*,k=t_1\) we get \eqref{equation: stopping time bound}.
\begin{align*}
    t_I&\leq t_S^{(u^*)}+t_I^{t_S^{(u^*)}+1}\\
    &\leq t_1+t_S^{t_1}-1+t_I^{t_S^{(u^*)}+1}.
\end{align*}
Then, for any \(f_1,\mathcal{J},t_1\) and \(X_1^{t_1-1}\),
\begin{align*}
    \left(t_I-t_1+1\right)^+&\leq \left(t_1+t_S^{t_1}-1+t_I^{t_S^{(u^*)}+1}-t_1+1\right)^+\\
    &= t_S^{t_1}+t_I^{t_S^{(u^*)}+1}
 \end{align*}
and
\begin{align*}
    E_{f_1,\mathcal{J},t_1}\left(\left(t_I-t_1+1\right)^+\middle|X_1^{t_1-1}\right)&\leq E_{f_1,\mathcal{J},t_1}\left(t_S^{t_1}+t_I^{t_S^{(u^*)}+1}\middle|X_1^{t_1-1}\right)\\
    &= E_{f_1,\mathcal{J},t_1}\left(t_S^{t_1}+t_I^{t_S^{(u^*)}+1}\right)\\
    &= E_{f_1,\mathcal{J},1}\left(t_S+t_I\right)
\end{align*}
since \(t_S^{(u^*)}\geq t_1\).
\end{proof}

\subsection{Proof of Lemma \ref{lemma: add}}
\begin{lemma*}
For any \(f_1,\mathcal{J}\) and \(\rho>0\), if \(t_1=1\) and \(c_n^D=0\) for \(n\leq(1+\rho)\frac{c^S}{\underline{q}}\), then
\begin{align*}
    E_{f_1,\mathcal{J},1}t_I&\leq \frac{c^S}{\underline{q}}
\end{align*}
as \(c^S\to\infty\) and \(\rho\to 0\).
\end{lemma*}
\begin{proof}
From Lemma \ref{lemma: conditional},
\begin{align}
    E_{f_1,\mathcal{J},1} t_I&= \frac{E_{f_1,\mathcal{J},1} t_S}{E_{f_1,\mathcal{J},1}\left(P_{f_1,\mathcal{J},1}\left(D_{t_S}\geq c_n^D\middle|t_S\right)\right)}. \label{equation: wadd fraction}
\end{align}
For the nominator of \eqref{equation: wadd fraction}, we consider the probability of \(t_S>t\) for each \(t\).
\begin{align}
    &P_{f_1,\mathcal{J},1}\left(t_S>t\right) \nonumber\\
    &= P_{f_1,\mathcal{J},1}\left(\max_{1\leq k\leq t}S_k<c^S\right) \nonumber\\
    &\leq P_{f_1,\mathcal{J},1}\left(t\left(q\left(\hat{f}_{X_1^t}\right)-\kappa\right)<c^S\right) \nonumber\\
    &\leq P_{f_1,\mathcal{J},1}\left(q\left(\hat{f}_{X_1^t}\right)-\sum_{j\in\mathcal{J}}q_{j,1}<\frac{c^S}{t}+\kappa-\sum_{j\in\mathcal{J}}q_{j,1}\right) \nonumber\\
    &\leq P_{f_1,\mathcal{J},1}\left(\left|q\left(\hat{f}_{X_1^t}\right)-\sum_{j\in\mathcal{J}}q_{j,1}\right|>\sum_{j\in\mathcal{J}}q_{j,1}-\kappa-\frac{c^S}{t}\right) \nonumber\\
    &\leq P_{f_1,\mathcal{J},1}\left(L\norm{\hat{f}_{X_1^t}-f_1}_1>\sum_{j\in\mathcal{J}}q_{j,1}-\kappa-\frac{c^S}{t}\right) \nonumber\\
    &= P_{f_1,\mathcal{J},1}\left(\norm{\hat{f}_{X_1^t}-f_1}_1>\frac{\sum_{j\in\mathcal{J}}q_{j,1}-\kappa-\frac{c^S}{t}}{L}\right) \nonumber\\
    &\leq 2^m\exp\left(-\frac{t}{2}\left(\frac{\sum_{j\in\mathcal{J}}q_{j,1}-\kappa-\frac{c^S}{t}}{L}\right)^2\right) \label{equation: wadd nominator}
\end{align}
where, in \eqref{equation: wadd nominator}, we have used the inequality in \cite{inequalities}. Then, for \(t>T=(1+\rho)\frac{c^S}{\underline{q}}\),
\begin{align*}
    \sum_{j\in\mathcal{J}}q_{j,1}-\kappa-\frac{c^S}{t}&\geq \sum_{j\in\mathcal{J}}q_{j,1}-\kappa-\frac{\underline{q}}{1+\rho}\\
    &\geq \frac{\rho\underline{q}}{1+\rho}
\end{align*}
since for all \(j\), \(q_{j,1}\geq\underline{q}_j>0\) and \(\mathcal{J}\neq\emptyset\). Thus,
\begin{align*}
    E_{f_1,\mathcal{J},1}t_S&= \sum_t P_{f_1,\mathcal{J},1}\left(t_S>t\right)\\
    &\leq T+\sum_{t>T} 2^m\exp\left(-t\frac{\rho^2\underline{q}^2}{2(1+\rho)^2L^2}\right)\\
    &= T+2^m\frac{\exp\left(-\frac{\rho^2\underline{q}c^S}{2(1+\rho)L^2}\right)}{1-\exp\left(-\frac{\rho^2\underline{q}^2}{2(1+\rho)^2L^2}\right)}
\end{align*}
Next, we bound the denominator of \eqref{equation: wadd fraction} using the fact that \(t_S\leq T\) with high probability. We also use the fact that whenever \(t_S\leq T\), \(n_{t_S}\leq T\) and \(D_{t_S}\geq c_{n_{t_S}}^D=0\) thus, \(t_I=t_S\). Then,
\begin{align*}
    &E_{f_1,\mathcal{J},1}P_{f_1,\mathcal{J},1}\left(t_I= t_S\middle|t_S\right)\\
    &\geq E_{f_1,\mathcal{J},1}\mathbb{1}\left(t_S\leq T\right)\\
    &\geq P_{f_1,\mathcal{J},1}\left(t_S\leq T\right)\\
    &= P_{f_1,\mathcal{J},1}\left(\max_{k\leq T}S_k\geq c^S\right)\\
    &\geq P_{f_1,\mathcal{J},1}\left(S_{ T}\geq c^S\right)\\
    &= P_{f_1,\mathcal{J},1}\left(q\left(\hat{f}_{X_1^T}\right)\geq\frac{\underline{q}}{1+\rho}\right)\\
    &= 1-P_{f_1,\mathcal{J},1}\left(\left|q\left(\hat{f}_{X_1^T}\right)-q\left(f_1\right)\right|>\sum_{j\in\mathcal{J}}q_{j,1}-\kappa-\frac{\underline{q}}{1+\rho}\right)\\
    &\geq 1-P_{f_1,\mathcal{J},1}\left(\norm{\hat{f}_{X_1^T}-f_1}_1>\frac{\min_j\underline{q}_j-\kappa-\frac{\underline{q}}{1+\rho}}{L}\right)\\
    &\geq 1-2^m\exp\left(-\frac{T}{2}\left(\frac{\rho\underline{q}}{(1+\rho)L}\right)^2\right)\\
    &\geq 1-2^m\exp\left(-\frac{\rho^2\underline{q}}{2(1+\rho)L^2}c^S\right)
\end{align*}
Finally,
\begin{align*}
    E_{f_1,\mathcal{J},1}t_I&\leq \frac{(1+\rho)\frac{c^S}{\underline{q}}+2^m\frac{\exp\left(-\frac{\rho^2\underline{q}}{2(1+\rho)L^2}c^S\right)}{1-\exp\left(-\frac{\rho^2\underline{q}^2}{2(1+\rho)^2L^2}\right)}}{1-2^m\exp\left(-\frac{\rho^2\underline{q}}{2(1+\rho)L^2}c^S\right)}
\end{align*}
and \(E_{f_1,\mathcal{J},1}t_I\leq\frac{c^S}{\underline{q}}\) as \(c^S\to\infty\) and \(\rho\to 0\).
\end{proof}

\bibliographystyle{IEEEtran}
\bibliography{bibliography.bib}

\end{document}